\documentclass{lmcs}
\pdfoutput=1

\usepackage{lastpage}
\lmcsdoi{16}{2}{12}
\lmcsheading{}{\pageref{LastPage}}{}{}%
{Jun.~17,~2019}{Jun.~18,~2020}{}

\usepackage[utf8]{inputenc}
\usepackage{amssymb}
\usepackage[numbers]{natbib}

\newcommand{\eqdef}{\stackrel{\rm def}{=}}

\newcommand\np{\mbox{NP}}

\newcommand\sz{\mbox{SIZE}}

\newcommand\tr{\mathsf{True}}

\newcommand\upci{{\mathsf{UB}^{i.o.}_{k}}}
\newcommand\lpv{{L({\mathsf{PV}})}}

\newcommand\un{{1^{(n)}}}
\newcommand\um{{1^{(m)}}}

\newcommand\pnp{\varphi_{\mathsf{P} = \mathsf{NP}}}
\newcommand\LB{\mathsf{LB}_k^{\textit{a.e.}}}

\begin{document}

\title{Consistency of circuit lower bounds with bounded theories}

\author[J.~Bydzovsky]{Jan Byd\v{z}ovsk\'{y}}	
\address{Institute of Discrete Mathematics and Geometry, Vienna University of Technology}	
\email{jan.bydz@gmail.com}  

\author[J.~Krajicek]{Jan Kraj\'{\i}\v{c}ek}	
\address{Faculty of Mathematics and Physics, Charles University in Prague}	
\email{krajicek@karlin.mff.cuni.cz}  

\author[I.C.~Oliveira]{Igor C.~Oliveira}	
\address{Department of Computer Science, University of Warwick}	
\email{igor.oliveira@warwick.ac.uk}

\begin{abstract}
Proving that there are problems in $\mathsf{P}^\mathsf{NP}$ that require boolean circuits of super-linear size is a major frontier in complexity theory. While such lower bounds are known for larger complexity classes, existing results only show that the corresponding problems are hard on \emph{infinitely many input lengths}. For instance, proving \emph{almost-everywhere} circuit lower bounds is open even for problems in $\mathsf{MAEXP}$. Giving the notorious difficulty of proving lower bounds that hold for all large input lengths, we ask the following question:
\begin{center}
\emph{Can we show that a large set of techniques cannot prove that $\mathsf{NP}$ is easy infinitely often?}
\end{center}
Motivated by this and related questions about the interaction between \emph{mathematical proofs} and \emph{computations}, we investigate circuit complexity from the perspective of logic.

Among other results, we prove that for any parameter $k \geq 1$ it is consistent with
theory $T$ that computational class ${\mathcal C} \not \subseteq \textit{i.o.}\sz(n^k)$,
where $(T, \mathcal{C})$ is one of the pairs:
\begin{center}
$T = \mathsf{T}^1_2$ and ${\mathcal C} = \mathsf{P}^\mathsf{NP}$, $\quad T = \mathsf{S}^1_2$ and ${\mathcal C} = \mathsf{NP}$, $\quad T = \mathsf{PV}$ and ${\mathcal C} = \mathsf{P}$.
\end{center}
In other words, these theories cannot establish infinitely often circuit upper bounds for the corresponding problems. This is of interest because the weaker theory $\mathsf{PV}$ already formalizes sophisticated arguments, such as a proof of the PCP Theorem \citep{Pich-pcp}. These consistency statements are unconditional and improve on earlier theorems of \citep{unprov} and \citep{BydMul} on the consistency of lower bounds with $\mathsf{PV}$.
\end{abstract}
\maketitle

\section{Introduction}\label{s:intro}

Understanding the computational power of polynomial size boolean circuits is one of the most mysterious questions in computer science. Despite major efforts to address this problem and significant progress in several \emph{restricted} settings (e.g.~\citep{DBLP:journals/siamcomp/Mulmuley99, rossmanphdthesis, DBLP:conf/stoc/MurrayW18}), it is consistent with current knowledge that every problem in $\mathsf{NP}$ can be computed by circuits containing no more than $4n$ gates \citep{DBLP:conf/focs/FindGHK16}. This bound is much weaker than the lower bound results conjectured by most (but not all) researchers in the field. For instance, it is reasonable to expect that computing $k$-clique on $n$-vertex graphs requires circuits of size $n^{\Omega(k)}$, but we appear to be very far from establishing a result of this form for \emph{unrestricted} boolean circuits.\\

\subsection*{Fixed-polynomial size circuit lower bounds.} Given the difficulty of proving stronger lower bounds for problems in $\mathsf{NP}$, a natural research direction is to investigate super-linear and fixed-polynomial circuit size lower bounds for problems in larger complexity classes. This line of work was started by Kannan \citep{Kan}, who showed that for each $k \geq 1$ there is a problem in
$\Sigma^p_2 \cap \Pi^p_2$ that cannot be computed by circuits of size $n^k$.
The result was subsequently improved by K\"{o}bler and Watanabe \citep{DBLP:journals/siamcomp/KoblerW98}, who obtained the same lower bound for the class $\mathsf{ZPP}^\mathsf{NP} \subseteq \Sigma^p_2 \cap \Pi^p_2$, and by Cai \citep{DBLP:journals/jcss/Cai07}, who showed it for $\mathsf{S}_2^p \subseteq \mathsf{ZPP}^\mathsf{NP}$. Two incomparable results were then obtained by Vinodchandran \citep{DBLP:journals/tcs/Vinodchandran05} and Santhanam \citep{DBLP:journals/siamcomp/Santhanam09}, who proved that $\mathsf{PP} \nsubseteq \mathsf{SIZE}[n^k]$ and $\mathsf{MA}/1 \nsubseteq \mathsf{SIZE}[n^k]$, respectively.\footnote{We use $\mathsf{SIZE}[s]$ to denote the set of languages computable by circuits of size at most $s(n)$ on every large enough input length. We say that a language $L$ is in $\textit{i.o.}\mathsf{SIZE}[s]$ if there is a language $L' \in \mathsf{SIZE}[s]$ such that $L$ and $L'$ agree on infinitely many input lengths.}

Modulo the use of a single bit of advice on each input length, Santhanam's lower bound is known to imply all aforementioned results. Unfortunately, there exist barriers to adapting his techniques to prove super-linear lower bounds for smaller classes such as $\mathsf{NP}$, as explained by Aaronson and Wigderson \citep{DBLP:journals/toct/AaronsonW09}. Establishing such lower bounds is also open for $\mathsf{P^{\mathsf{NP}}}$, and constitutes an important frontier in the area of fixed-polynomial size lower bounds.\footnote{Indeed, a proof that $\mathsf{E}^\mathsf{NP} \nsubseteq \mathsf{SIZE}[n^{1.01}]$ would be considered a breakthrough by some researchers in the field.} Interestingly, it is known that proving that $\mathsf{P}^\mathsf{NP} \nsubseteq \mathsf{SIZE}[n^k]$ for all $k$ is equivalent to showing a stronger Karp-Lipton collapse under the assumption that $\mathsf{NP} \subseteq \mathsf{SIZE}[\mathsf{poly}]$ \citep{CMMW19}.\footnote{Some of our consistency results can be interpreted from this perspective: it is possible to establish stronger ``logical'' Karp-Lipton collapses if $\mathsf{NP} \subseteq \mathsf{SIZE}[\mathsf{poly}]$ and this inclusion is \emph{provable} in certain theories \citep{CooKra}.} We refer to \citep{DBLP:journals/jco/CaiC06, DBLP:books/sp/goldreich2011/GoldreichZ11} for more information about uniform complexity classes around $\mathsf{P}^\mathsf{NP}$.

While existing circuit lower bounds might not be entirely satisfactory from the perspective of the uniform complexity of the problems, there is another important issue with these results: they only establish  hardness on \emph{infinitely many input lengths}. Could it be the case that some natural problems are easy on some input lengths and hard on others? Perhaps the existence of exceptional mathematical structures\footnote{In the sense of \url{https://en.wikipedia.org/wiki/Exceptional_object}} of certain sizes might affect (non-uniform) complexity theory around some input lengths? This possibility seems unlikely, but we are far from understanding the situation. For instance, a basic question in complexity that remains open is whether the nondeterministic time-hierarchy theorem can be extended to an almost-everywhere result (see \citep{DBLP:conf/icalp/BuhrmanFS09}). On the algorithmic side, an intriguing example is that the natural problem of generating canonical prime numbers admits a faster algorithm on infinitely many input lengths \citep{DBLP:conf/stoc/OliveiraS17}, but showing that the algorithm succeeds on all input lengths is open. More recent works such as  \citep{DBLP:journals/iandc/FortnowS17} and \citep{DBLP:conf/stoc/MurrayW18} show that quite often \emph{some} control can be obtained over the set of hard input lengths. Still, proving an \emph{almost-everywhere} circuit size lower bound beyond $4n$ gates remains open even for problems in $\mathsf{MATIME}[2^{n}]$ (see \citep{DBLP:conf/coco/BuhrmanFT98} for a related lower bound).\\

Addressing these questions without further assumptions (i.e.~unconditionally) appears to be extremely challenging. In this work, we attempt to provide \emph{formal evidence} that some problems in \emph{lower uniform complexity classes} are hard on \emph{every large enough input length}. This can be done via the investigation of circuit complexity from the perspective of mathematical logic. More precisely, we are interested in \emph{unconditional} results showing that lower bounds such as $\mathsf{NP} \nsubseteq \textit{i.o.}\mathsf{SIZE}[n^3]$ are \emph{consistent} with certain logical theories.\footnote{Note also that establishing a consistency statement is a \emph{necessary} step before the corresponding circuit lower bound can be unconditionally established, since a true statement is always consistent with a sound theory.} To obtain interesting results, we consider theories that can formalize a variety of  techniques from algorithms, complexity, and related areas. We  focus on first-order theories in the standard sense of mathematical logic, which offers a principled way of investigating consistency statements of the form above. We describe next the theories relevant to our work.\\

\subsection*{Bounded Arithmetic.} Bounded arithmetic theories are fragments of Peano Arithmetic with close connections to computational complexity and proof complexity. Such theories have been widely investigated by logicians and complexity theorists since the 1970's. Among the most influential theories we have
Cook's equational theory $\mathsf{PV}$ \cite{Coo75} and its corresponding first-order formalization \citep{KPT} (see also \citep{jerabek:sharply-bounded}),\footnote{In this paper we use $\mathsf{PV}$ to refer to its first-order formulation (cf. \citep[Section 5.3]{kniha}).} Buss's theories $\mathsf{S}^1_2$ and $\mathsf{T}^1_2$ \cite{Bus-book}, and extensions of these theories by variants of the pigeonhole principle developed primarily by Je\v r\' abek \cite{Jer04,Jer-phd,Jer07,Jer09} (such as theory $\mathsf{APC}_1$ extending $\mathsf{PV}$). The objects of study in these theories are natural numbers (representing finite binary strings),
and the basic functions and relations are given by polynomial-time ($p$-time) algorithms in some programming scheme. For instance,
Cook \cite{Coo75} relied on Cobham's theorem \cite{Cob64} that all $p$-time functions can be generated from
few initial ones by composition and bounded recursion on notation.
For convenience, the language $\lpv$ we adopt here is the language of $\mathsf{PV}$,
having a function symbol for each $p$-time algorithm.\footnote{\label{f:AKS}This does not necessarily imply that $\mathsf{PV}$ can prove the relevant properties of its function symbols. For instance, the $\mathsf{AKS}$ algorithm \citep{Agrawal02primesis} for testing primality appears as some symbol $f_{\mathsf{AKS}} \in L(\mathsf{PV})$, but $\mathsf{PV}$ might not be able to prove that $x$ is prime if and only if $f_\mathsf{AKS}(x) = 1$.} There are relation
symbols $=$ and $\le$ with their usual meaning, and all other relations we want to include are
represented by their characteristic functions. The specific axiomatization of $\mathsf{PV}$ is not important here: everything
will also work for the theory of all true universal $\lpv$-sentences (to be denoted by $\tr_0$),
and $\mathsf{PV} \subseteq \tr_0$. We only note that $\mathsf{PV}$ proves induction for all $p$-time predicates
by formalizing binary search, cf.~\cite{KPT,kniha}.

The original language of theories $S^1_2$ and $T^1_2$ as defined in \cite{Bus-book} is a finite subset of $\lpv$, but we consider theories $S^1_2(\mathsf{PV})$ and $T^1_2(\mathsf{PV})$ in the richer language $\lpv$. (We will add to these theories even more axioms, which makes any consistency statement stronger.) The principal axioms of the two theories are length-induction (LIND) and
induction (IND), respectively, accepted for $\Sigma^b_1(\mathsf{PV})$-formulas.\footnote{We review later in the text some definitions necessary in this work. For more information about standard concepts in bounded arithmetic, we refer to a reference such as \citep{kniha}.}
Theory $S^1_2(\mathsf{PV})$ is close to $\mathsf{PV}$ (it is $\forall \Sigma^b_1(\mathsf{PV})$-conservative over it),
but $T^1_2(\mathsf{PV})$ appears to be significantly stronger (cf.~\cite{kniha}).
Theories  $\mathsf{PV}$, $\mathsf{S}^1_2(\mathsf{PV})$ and $\mathsf{T}^1_2(\mathsf{PV})$ and their
extensions by a form of the pigeonhole principle (often referred to as $\mathsf{dWPHP}$ or $\mathsf{sWPHP}$)
are actually quite strong for the purposes of complexity theory. They are now known to formalize many
key theorems in algorithms, combinatorics, complexity, and related fields (cf.~\cite{PW,Bus-book,kniha, razborov1995bounded, Jer-phd,Jer04,Jer07,Jer09,CooNgu,Pich-phd,Pich,Pich-pcp, buss2015collapsing, DBLP:journals/corr/abs-1103-5215, ojakian2004combinatorics, DTML_thesis, DBLP:journals/eccc/MullerP17}
and references therein).

Recall that the class of $\Sigma^b_1(\mathsf{PV})$-formulas consists of formulas of the form
\[
\exists y_1 \le t_1(\overline x)\dots \exists y_k \le t_k(\overline x)\,
A(\overline x, \overline y)\;,
\]
where the $t_i$ are $\lpv$-terms not involving $y_i$, and $A$ is quantifier free. The definition of this class
in the original language of $S^1_2$ is a bit more complicated (distinguishing two kinds of bounded quantifiers),
but in our language $\lpv$ it is equivalent to this simpler definition.
The class of $\Sigma^b_2(\mathsf{PV})$-formulas is defined similarly, but the formula $A$ can also be the negation of a
$\Sigma^b_1(\mathsf{PV})$-formula (these negations are $\Pi^b_1(\mathsf{PV})$-formulas).
The predicates definable over the natural numbers by $\Sigma^b_1(\mathsf{PV})$-formulas and by $\Sigma^b_2(\mathsf{PV})$-formulas
are exactly the predicates from $\Sigma^p_1 = \np$ and from $\Sigma^p_2$, respectively. We shall denote the theory of all true $\forall \Sigma^b_1(\mathsf{PV})$-sentences by $\tr_1$.\\

\subsection*{Our results.} For an $\lpv$-formula $\varphi(x)$ and an integer $k \geq 1$, the $\lpv$-sentence $\upci(\varphi)$ is defined as follows:
\begin{equation}\label{eq:defformula}
\forall 1^{(n)} \,\exists \um (m \geq n)\,\exists C_m (|C_m| \le m^k) \,\forall x (|x|=m),\
\varphi(x) \equiv (C_m(x)=1)\ .
\end{equation}
The sentence $\upci(\varphi)$ formalizes that the $m$-bit boolean functions defined by $\varphi$ (over different input lengths) are computed infinitely often (\emph{i.o.}) by circuits of size $m^k$.\footnote{The notation $1^{(n)}$ means that $n$ is the length of another variable. We abuse notation and use $|C_m|$ to denote the number of gates in $C_m$. We refer to \citep{Pich} for a detailed discussion of the formalization of circuit complexity in bounded arithmetic.}

We unconditionally establish that almost-everywhere circuit lower bounds for complexity classes contained in $\mathsf{P}^{\mathsf{NP}}$ are consistent with bounded arithmetic theories.

\begin{thm}[Consistency of almost-everywhere circuit lower bounds with bounded \nolinebreak theo\nolinebreak ries] \label{t:main} Let $k \geq 1$ be any positive integer. For any of the following pairs of an $\lpv$-theory $T$ and a uniform complexity class ${\mathcal C}$:
\begin{enumerate}[label=(\alph*)]

\item $T = \mathsf{T}^1_2(\mathsf{PV}) \cup \tr_1$ and ${\mathcal C} = \mathsf{P}^\mathsf{NP}$,

\item $T = \mathsf{S}^1_2(\mathsf{PV}) \cup \tr_0$ and ${\mathcal C} = \mathsf{NP}$,

\item $T = \mathsf{PV} \cup \tr_0$ and ${\mathcal C} = \mathsf{P}$,

\end{enumerate}
there is an $\lpv$-formula $\varphi(x)$ defining a language $L \in \mathcal C$
such that
$T$ does not prove the sentence $\upci(\varphi)$.
\end{thm}

Our arguments are somewhat non-constructive and do not provide a single explicit formula $\varphi(x)$ in each case of the result. Informally, Theorem \ref{t:main} shows (in particular) the following consistency statements:
\begin{align*}
  T^1_2(\mathsf{PV}) \,&\nvdash\, \mathsf{P}^\mathsf{NP}\subseteq \text{\emph{i.o.}}\mathsf{SIZE}[n^k]\\
  S^1_2(\mathsf{PV}) \,&\nvdash\, \mathsf{NP} \subseteq \text{\emph{i.o.}}\mathsf{SIZE}[n^k]\\
  \mathsf{PV} \,&\nvdash\, \mathsf{P} \subseteq \text{\emph{i.o.}}\mathsf{SIZE}[n^k]
\end{align*}

In other words, there are models of these theories (satisfying a large fraction of modern complexity theory) that contain explicit problems that require circuits of size $n^k$ on every large enough input length.\footnote{A bit more precisely, the lower bound holds for every input length $n \geq n_0$, where $n_0$ is an element of the model. Note that $n_0$ might be a nonstandard element of this model.} Another interpretation is that one can develop theories of computational complexity that postulate the existence of hard problems (as new axioms) without ever proving a contradictory statement.\footnote{One can even contemplate the possibility that more advanced consistency results might allow the development of ``logic-based'' cryptography: protocols that are unconditionally secure against all efficient algorithms that can be proved correct in a given theory.} As alluded to above, given the expressive power of these theories, we view the consistency results as evidence that such lower bounds hold in the standard mathematical universe. Nevertheless, if one strongly believes in an inclusion such as $\mathsf{NP} \subseteq \mathsf{SIZE}[n^k]$ for a large enough $k$, then Theorem \ref{t:main} shows that even to prove this inclusion on infinitely many input lengths it will be necessary to use mathematical arguments that are beyond the reasoning capabilities of the corresponding theories.

We stress that $\mathsf{True}_0$ and $\mathsf{True}_1$ contain several statements of interest about algorithms, boolean circuits, extremal combinatorial objects, etc.~Theorem \ref{t:main} shows that even assuming such statements as  axioms the corresponding theories cannot prove fixed-polynomial size circuit upper bounds.\footnote{For instance, $\mathsf{T}^1_2(\mathsf{PV}) \cup \mathsf{True}_1$ proves the correctness of the AKS primality testing algorithm (see Footnote \ref{f:AKS}), i.e., it shows that $\forall x \, (\exists y\,(1 < y < x \wedge y \,|\, x) \leftrightarrow f_\mathsf{AKS}(x) = 0)$ since this sentence is in $\mathsf{True}_1$. This implies that this theory proves that primality testing can be done by circuits of size $n^c$ for a fixed $c$ on  every large enough input length $n$.}

We note that the particular syntactic form of $\varphi$ defining the hard language in Theorem \ref{t:main} items (a) and (c) is irrelevant as long as $p$-time functions and predicates are defined
by open formulas of the language of $\mathsf{PV}$ and the SAT predicate used in the argument is defined
by a $\Sigma^b_1$-formula. Indeed, if two open $L(\mathsf{PV})$-formulas
 define the same predicate then this universal statement
is included in theory $\mathsf{True}_0$ and hence (c) holds identically for
all open formulas defining the same language.
An analogous observation applies to (a): languages in $\mathsf{P}^{\mathsf{NP}}$ are definable by $\Delta^b_2$-formulas w.r.t.~the theory\footnote{In other words, the languages have both $\Sigma^b_2$ and $\Pi^b_2$ definitions that are
provably equivalent in the theory.}  and
the universal statement stating their equivalence is thus in $\mathsf{True}_1$.\\

\subsection*{Related work and techniques.} Some works have investigated the \emph{unprovability} of circuit lower bounds, or equivalently, the \emph{consistency of upper bounds}. We refer to the introduction of \citep{DBLP:journals/eccc/MullerP17} for more information about this line of work, and to Appendix \ref{s:appendix_unprov} for some related remarks that might be of independent interest. Theorem \ref{t:main} and our techniques are more directly connected to \citep{CooKra}, \citep{unprov}, and \citep{BydMul}. We review the relevant results next.

Cook and Kraj\'{\i}\v{c}ek \citep{CooKra} (see also \citep{krajicek_extensions}) were the first to systematically investigate the consistency of circuit lower bounds. They established several results showing that $\mathsf{NP} \nsubseteq \mathsf{SIZE}[\mathsf{poly}]$ is consistent with $\mathsf{PV}$, $\mathsf{S}^1_2$, and $\mathsf{T}^1_2$ under appropriate assumptions regarding the collapse of $\mathsf{PH}$. For instance, it was shown (in particular) that $\mathsf{T}^1_2 \nvdash \mathsf{NP} \subseteq \mathsf{SIZE}[\mathsf{poly}]$ if $\mathsf{PH} \nsubseteq \mathsf{P}^\mathsf{NP}$.  While their results are \emph{conditional}, \citep{CooKra} considered consistency statements for a fixed language in $\mathsf{NP}$ with respect to all polynomial bounds. In \citep{unprov}, two of the authors established an \emph{unconditional} result showing that $\mathsf{PV} \nvdash \mathsf{P} \subseteq \mathsf{SIZE}[n^k]$, where $k$ is any fixed integer. This consistency statement was subsequently improved by \citep{BydMul}, who considered a more natural formalization of the statement that a language has circuits of size $O(n^k)$ and adapted the argument \citep{unprov} using polynomial-time ultrapowers.\footnote{The formalizations in \citep{unprov}  and \citep{BydMul} differ on how the $O(\cdot)$ notation is handled, and we refer to the corresponding papers for details. Here the sentences $\upci(\varphi)$ refer to infinitely often upper bounds, and this issue is not relevant.} All previous results refer to the consistency of lower bounds on infinitely many input lengths, and Theorem \ref{t:main} part (c) strictly improves upon \citep{unprov} and \citep{BydMul}.\footnote{In model-theoretic terms, \citep{unprov} and \citep{BydMul} provide models where the circuit lower bound holds on some large enough input length. A
slight modification of the proof in \citep{BydMul} gives a fixed model with arbitrarily large hard
input lengths (Moritz M\"{u}ller, private communication). On the other hand, our results provide a model where the lower bound holds on every large enough input length.}

In terms of techniques, the proof of Theorem \ref{t:main} explores methods from complexity theory and mathematical logic to establish the  unprovability of infinitely often upper bounds. We combine ideas from the conditional results of \citep{CooKra} with the unconditional approach of \citep{unprov}. The general theme is to obtain computational information from proofs in the corresponding bounded theories. For instance, under the assumption that there is a $\mathsf{PV}$-proof $\pi$ that a problem in $\mathsf{P}$ admits \emph{non-uniform} circuits of size $n^k$, we attempt to extract from $\pi$ a more ``uniform'' construction of such circuits. The ideal plan is to contradict existing lower bounds against uniform circuits, such as those investigated in \citep{SW} and other works. However, as explained in \citep{unprov}, implementing this plan is not straightforward, since the ``uniformity'' one obtains from $\pi$ does not match existing results in the area of uniform circuit lower bounds. Moreover, the proof of Theorem \ref{t:main} creates additional difficulties because the uniform circuit lower bounds, already insufficient, only hold on infinitely many input lengths. In order to overcome this difficulty, we make use of further insights on the logical side of the argument. In turn, this requires appropriate extensions of the complexity-theoretic arguments.\\

\subsection*{Extensions and open problems.} One can adapt the methods used in the proof of Theorem \ref{t:main} to show that $\mathsf{APC}^1$ and indeed theory
\begin{equation}\label{eq:unprov_APC}
\mathsf{S}^1_2(\mathsf{PV}) \cup \mathsf{sWPHP}(\mathsf{PV}) \;\nvdash\; \upci(\varphi)\; ,
\end{equation} for some $L(\mathsf{PV})$-formula $\varphi(x)$ defining a language in $\mathsf{ZPP}^{\mathsf{NP}[O(\log n)]}$.\footnote{This is obtained as in Theorem \ref{t:main} parts (a) and (b) by proving the following ``logical'' Karp-Lipton collapse: If $\mathsf{S}^1_2(\mathsf{PV}) \cup \mathsf{sWPHP}(\mathsf{PV}) \vdash \mathsf{NP} \subseteq \mathsf{SIZE}[\mathsf{poly}]$ then $\mathsf{PH}$ collapses to $\mathsf{ZPP}^{\mathsf{NP}[O(\log n)]}$. The proof of the latter adapts the argument in \citep[Theorem 5.1 (\emph{ii})]{CooKra}, using randomization to obtain witnesses for the required $\mathsf{dWPHP}$ axioms and an $\mathsf{NP}$ oracle to check that they are correct. (A bit more formally, the idea is to first Skolemize the theory, reducing the argument to the case of $\mathsf{S}^1_2(\mathsf{PV})$, then to handle the newly introduced function symbols by witnessing them in the standard model through a probabilistic computation with an $\mathsf{NP}$ oracle.)} In contrast, existing (infinitely often) lower bounds for $\mathsf{ZPP}^\mathsf{NP}$ seem to hold only when the $\mathsf{NP}$ oracle is adaptively queried polynomially many times \citep{DBLP:journals/siamcomp/KoblerW98, DBLP:journals/jcss/Cai07}, or with respect to non-adaptive queries but for a promise version of this class (see the discussion in  \citep[Section 3.2]{DBLP:journals/siamcomp/Santhanam09}). There is strong evidence that asking more queries increases computational power (see \citep{DBLP:journals/jco/CaiC06, DBLP:conf/coco/ChangP08} and references therein), and it is known that polynomially many non-adaptive queries to an $\mathsf{NP}$ oracle are equivalent in power to logarithmic many adaptive queries \citep{DBLP:journals/jcss/Hemachandra89, DBLP:journals/iandc/BussH91}. The problem of proving super-linear circuit lower bounds for $\mathsf{ZPP}^\mathsf{NP}_\mathsf{tt}$ (i.e.~$\mathsf{ZPP}^\mathsf{NP}$ with non-adaptive queries) was investigated recently by \citep{DBLP:conf/mfcs/DixonPV18}, and in a sense the consistency statement in (\ref{eq:unprov_APC}) addresses this question with respect to $\mathsf{APC}^1$.

On the one hand, this consistency statement feels less appealing than the results in Theorem \ref{t:main} due to its proximity to existing lower bounds in complexity theory. But on the other hand, it highlights the importance of $\mathsf{APC}^1$ in connection to frontier questions in complexity theory and lower bounds. As one of our main open problems, we ask for the proof of stronger consistency results for the theory $\mathsf{APC}^1$. For instance, can one show that $\mathsf{APC}^1 \nvdash \mathsf{MA} \subseteq \mathsf{SIZE}[n^k]$, partially addressing the use of non-uniform advice in \citep{DBLP:journals/siamcomp/Santhanam09}? In connection to this and related problems, it might be fruitful to investigate a potential extension of the equivalence in \citep{CMMW19} to a result that relates consistency statements, witnessing theorems, and logical Karp-Lipton theorems.

It would also be interesting to improve our consistency results for $\mathsf{S}^1_2$ and $\mathsf{T}^1_2$ with respect to the uniformity of the hard problems, and to establish a non-trivial statement about the consistency of circuit lower bounds with $\mathsf{T}^2_2$ (Theorem \ref{t:main} part (\emph{a}) extends to $\mathsf{S}^2_2$ using a similar argument and appropriate results from \citep{CooKra}).\\

We include in Appendix \ref{s:appendix_unprov} a discussion on the consistency of $\mathsf{P} \neq \mathsf{NP}$ and its connection to the unprovability of circuit lower bounds.

\section{Background and notation}

In order to emphasize the main ideas, we assume some familiarity with logic, bounded arithmetic, and complexity theory. Everything needed can be found in \cite{kniha}. The interested reader can consult \citep{CooNgu} for a more recent reference in bounded arithmetic, \citep{buss1997bounded} for a concise introduction, and \citep{DBLP:books/daglib/0031527} for an accessible exposition. For more background in circuit complexity, we refer to \citep{DBLP:books/daglib/0028687}. For a discussion of the formalization of complexity theory and circuit complexity in bounded arithmetic, see \citep{DBLP:journals/eccc/MullerP17} and references therein.

Our proofs will rely on some results and arguments from \citep{CooKra} and \citep{unprov}, and we refer to the detailed presentation in these papers instead of repeating the proofs here. In more detail, what is needed  from \citep{CooKra} is that some or their theorems can be modified to include  $\mathsf{True}_0$ or $\mathsf{True}_1$. On the other hand, the proof of Theorem \ref{t:main} (c) can only be followed if the reader is familiar with the simpler argument from \citep{unprov}.

We use $\mathsf{P}^{\mathsf{NP}[\ell(n)]}$ to denote the set of languages decided by a deterministic polynomial time machine that makes at most $\ell(n)$ queries to an $\mathsf{NP}$ oracle. We will assume without loss of generality that the oracle is some fixed $\mathsf{NP}$-complete language such as formula satisfiability.

\section{Consistency of lower bounds with bounded arithmetic}

This section proves Theorem \ref{t:main}. Let $k$ be a positive integer. We argue in each item as follows.\\

\noindent (\emph{a}) We consider two cases. If the polynomial hierarchy $\mathsf{PH}$ collapses to $\mathsf{P}^\mathsf{NP}$, then we can define a language $L \in \mathsf{P}^\mathsf{NP}$ such that $L \notin \textit{i.o.}\mathsf{SIZE}[n^k]$. More precisely, $L$ computes on input length $n$ as the lexicographic first truth-table corresponding to a function $h \colon \{0,1\}^n \to \{0,1\}$ that cannot be computed by circuits of size $n^k$. This language can be easily specified using a constant number of quantifiers over strings of length $\mathsf{poly}(n)$ (cf.~\citep{Kan}). By the equivalence between languages in $\Sigma^p_i$ and predicates definable by $\Sigma^b_i(\mathsf{PV})$ formulas (see e.g.~\citep[Theorem 3.2.12]{kniha}), there is an $L(\mathsf{PV})$-formula $\varphi_L(x)$ that defines $L$ (using the correspondence between $\{0,1\}^*$ and $\mathbb{N}$). Since $\mathsf{T}^1_2(\mathsf{PV}) \cup \mathsf{True}_1$ is sound and $L$ is hard on every large enough input length, this theory cannot prove the sentence $\upci(\varphi_L)$.

Assume now that $\mathsf{PH}$ does not collapse to  $\mathsf{P}^\mathsf{NP}$. Let $\varphi_\mathsf{SAT}(x)$ be a $\Sigma^b_1(\mathsf{PV})$-formula that defines the formula satisfiability problem (SAT). We take a particular formulation of $\varphi_\mathsf{SAT}(x)$ for which the input encoding is paddable, meaning that inputs of the satisfiability problem of length $\ell < m$ can be easily converted into equivalent inputs of length $m$. If $\mathsf{T}^1_2(\mathsf{PV}) \cup \mathsf{True}_1$ does not prove $\upci(\varphi_\mathsf{SAT})$, we are done, given that this formula defines a language in $\mathsf{NP} \subseteq \mathsf{P}^\mathsf{NP}$. Suppose $\mathsf{T}^1_2(\mathsf{PV}) \cup \mathsf{True}_1 \vdash \upci(\varphi_\mathsf{SAT})$. This formula has unbounded existential quantifiers, but since $\mathsf{T}^1_2(\mathsf{PV}) \cup \mathsf{True}_1$ is axiomatized by bounded formulas, Parikh's theorem (cf.~\citep[Section 5.1]{kniha}) implies that there is an $\lpv$-term $t(x)$ such that $\upci(\varphi_\mathsf{SAT})$ is provable in the theory even if
the existential quantifiers are bounded by $t(\un)$. In particular, $m$ and $|C_m|$ in Equation (\ref{eq:defformula}) can be bounded as
\begin{equation} \label{31.3.19b}
m, \ |C_m|\ \le\ n^{O(1)}\ .
\end{equation}
By our assumption on paddability, a circuit $C_m$ deciding satisfiability on formulas encoded using $m$ bits also works for all formulas of length $n \le m$.
But by (\ref{31.3.19b}), $|C_m| \le n^{O(1)}$ and hence $C_m$ can serve as a polynomial size circuit solving SAT on
formulas of size $n$. Consequently, if $\mathsf{T}^1_2(\mathsf{PV}) \cup \mathsf{True}_1$ proves that SAT is infinitely often in $\mathsf{SIZE}[n^k]$ it also proves that
SAT is in $\mathsf{SIZE}[\mathsf{poly}(n)]$. We now invoke the argument of \citep[Theorem 5.1 (\emph{iii})]{CooKra} who showed (in particular) that if $\mathsf{T}^1_2$ proves that SAT~$\in \mathsf{SIZE}[\mathsf{poly}]$ then $\mathsf{PH}$ collapses to $\mathsf{P}^\mathsf{NP}$. Their proof can be adapted to $\mathsf{T}^1_2(\mathsf{PV}) \cup \mathsf{True}_1$, since all sentences in $\tr_1$ are witnessed by $\mathsf{FP}^\mathsf{NP}$ functions and adding these sentences as new axioms does not affect the required witnessing theorem.\footnote{In more detail, adding a function symbol for these witnessing
functions turns sentences from $\mathsf{True}_1$ into universal sentences,
and universal sentences do not influence witnessing theorems. For example, if $\forall x \exists y\;(y \leq s(x) \wedge A(x,y))$ is in $\mathsf{True}_1$ and $f$ is the symbol for the associated witnessing function,
the universal sentence will be $\forall x\;(f(x) \le s(x) \wedge A(x,f(x)))$.} This collapse of $\mathsf{PH}$ is in contradiction to our assumption in this case of the proof, which completes the argument.\\

\noindent (\emph{b}) Consider the formula $\varphi_\mathsf{SAT}(x)$ defined in item (\emph{a}) above. If $\mathsf{S}^1_2 \cup \mathsf{True}_0 \nvdash \upci(\varphi_\mathsf{SAT})$ there is nothing else to prove. Otherwise, by the same argument via Parikh's Theorem it follows that $\mathsf{S}^1_2 \cup \mathsf{True}_0$ proves that SAT admits polynomial size circuits on every input length. Now the argument in \citep[Theorem 5.1 (\emph{ii})]{CooKra} (easily modifiable to handle $\mathsf{True}_0$ because axioms in it are universal sentences) implies that every language $L \in \mathsf{PH}$ is also in $\mathsf{P}^{\mathsf{NP}[c \cdot \log n]}$ for some $c \in \mathbb{N}$. In particular, every such language is in $\mathsf{P}^{\mathsf{NP}[n]}$. Consequently, by Kannan's construction \citep{Kan} there is a language $L_\mathsf{hard} \in \mathsf{P}^{\mathsf{NP}[n]}$ such that $L_\mathsf{hard} \notin \textit{i.o.}\mathsf{SIZE}[n^{k + 2}]$. We will need the following lemma.\footnote{The proof of the lemma uses ideas from the proof of \cite[Proposition 1.3]{fragm} showing that $\mathsf{S}^1_2$ can define all $\mathsf{FP}^{\mathsf{NP}[wit,O(\log n)]}$ functions, extending a proof from \cite{Bus-book} that $\mathsf{T}^1_2$ can define all $\mathsf{FP}^\mathsf{NP}$ functions. A similar argument was employed in the proof of \cite[Theorem 10]{FSW09} (without the infinitely often condition).}

\begin{lem}\label{l:PNP_to_NP}
If $\mathsf{NP} \subseteq \textit{i.o.}\mathsf{SIZE}[n^k]$ then $\mathsf{P}^{\mathsf{NP}[n]} \subseteq \textit{i.o.}\mathsf{SIZE}[n^{k +2}]$.
\end{lem}

\begin{proof} Let $L$ be a language in $\mathsf{P}^{\mathsf{NP}[n]}$ decided by a deterministic polynomial-time oracle machine $M$ running in time at most $q(n)$. For convenience, we assume without loss of generality that $M$ makes \emph{exactly} $n$ queries before accepting or rejecting an input string, regardless of the answers provided by its $\mathsf{NP}$ oracle $O$.

We consider the language $L_\mathsf{aux}$ containing all tuples $(a,j,b_1, \ldots, b_n, 1^{(t)}, c)$, where $|a| = n$, $1 \leq j \leq n$, each $b_i \in \{0,1\}$, $t \in \mathbb{N}$ is a padding parameter, and $c \in \{0,1\}$ is a control bit, which satisfy the following conditions:
\begin{itemize}
\item If $c = 0$, then when $M$ computes on $a$ and its first $j$ queries are answered according to $b_1, \ldots, b_j$, for each $i \leq j$ if $y_i \in \{0,1\}^\star$ is the $i$-th query and $b_i = 1$ we have $y_i \in O$.
\item If $c = 1$, the machine $M$ accepts $a$ within $q(n)$ steps under oracle answers $b_i$ for $1 \leq i \leq n$.
\end{itemize}
Since $O \in \mathsf{NP}$ and $M$ is a deterministic polynomial time machine, $L_\mathsf{aux} \in \mathsf{NP}$. Using the hypothesis of the lemma, for infinitely many values of $n$ there exists $t \leq 10n$ and a circuit $D_n$ for $L_\mathsf{aux}$ of size at most $C(n + \log n + n + t + 1)^k \leq n^{k + 1}$ (for large enough $n$) that decides $L_\mathsf{aux}$ with respect to our parameter $n$ (the input length for an instance of $L$). Note that the parameter $t$ allows us to hit the ``good'' input lengths without technical considerations about the input encoding employed in the definition of $L_\mathsf{aux}$.

We will use $D_n$ (with the correct value $t$ non-uniformly hardcoded in the input) as a sub-routine in order to solve $L$ on inputs of length $n$, as described next. First, we recover the correct oracle answers $d_1, \ldots, d_n$ for a given input string $a$. This is done in $n$ steps, where the $i$-th step recovers $d_i$. To recover $d_1$, we use $D_n$ to compute
\[
D_n(a,\stackrel{j}{1},\overbrace{1, \star, \ldots, \star}^{\vec{b}}, 1^{(t)}, \stackrel{c}{0}),
\]
where each $\star$ can be replaced by an arbitrary bit. If the output is $1$, the first query made by $M$ on $a$ has a positive answer with respect to $O$ (since positive queries must be strings in $O$ by the definition of tuples in $L_\mathsf{aux}$ when $c = 0$). Otherwise, we must have $d_1 = 0$. Next, we invoke
\[
D_n(a,\stackrel{j}{2},\overbrace{d_1, 1, \star, \ldots, \star}^{\vec{b}}, 1^{(t)}, \stackrel{c}{0}),
\]
knowing that the answer to the first query is correct. By the same argument, we are able to recover $d_2$, and proceeding similarly, we can recover all correct answers $d_1, \ldots, d_n$. Finally, by invoking $D_n(a,n,d_1, \ldots, d_n,1^{(t)},1)$ with $c=1$ and using the correct oracle answers, we can decide if $a \in L$. Clearly, this entire computation can be performed by a circuit of size at most $O(n \cdot |D_n|) = O(n^{k + 2})$, which completes the proof.\end{proof}

It follows from Lemma \ref{l:PNP_to_NP} and the properties of $L_\mathsf{hard}$ that there is a language $L \in \mathsf{NP}$ such that $L \notin \textit{i.o.}\mathsf{SIZE}[n^k]$. Consequently, if $\varphi_L(x)$ is a formula that defines $L$ then $\mathsf{S}^1_2 \cup \mathsf{True}_0 \nvdash \upci(\varphi_L)$. This completes the proof of item (\emph{b}).\\

\noindent (\emph{c}) We follow the overall strategy of the proof of \cite[Theorem 2.1]{unprov} (which combines the proof of \citep[Theorem 1.1]{SW} with other ideas), but the infinitely often statement considered here introduces certain difficulties. In particular, it is not clear how to adapt the proof in \citep{SW} to show that $\mathsf{P}$ is not contained infinitely often in $\mathsf{P}$-uniform $\mathsf{SIZE}[n^k]$. In general, combining different computations that succeed infinitely often might not produce a computation that succeeds infinitely often. We explain below how the argument from \citep{unprov} can be modified to establish the stronger statement in part (\emph{c}). (For simplicity of notation, we restrict our discussion to $\mathsf{PV}$, but the argument works for $\mathsf{PV} \cup \mathsf{True}_0$ as well.)

Let $g_{k'}$ for $k' = 3k$ be the $\mathsf{PV}$ function symbol provided by \citep[Lemma 3.1]{unprov}. Recall that $\mathsf{PV}$ proves that any uniform algorithm $h$ running in time at most $n^{k'-1}$ will fail to compute $g_{k'}$, even if $h$ is given a certain amount of advice that can depend on the input length. If $\mathsf{PV} \nvdash \upci(g_{k'})$ we are done. Otherwise, applying the KPT Theorem (see e.g.~\citep[Theorem 4.1]{unprov}) to sentence $\upci(g_{k'})$ (note crucially that $\upci(g_{k'})$ has the right quantifier complexity), we obtain a fixed $r \in \mathbb{N}$ (independent of $n$) and $\mathsf{PV}$ function symbols $f_1, \ldots, f_r$ such that on input  $1^{(n)}$ each function $f_i$ outputs $n \leq n_i \leq n^{a_i}$ (represented as $1^{(n_i)}$) and a circuit $C^{i}_{n_i}$ of size at most $n_i^{k}$ that is a candidate circuit for $g_{k'}$ on inputs of length $n_i$ (the upper bound $n^{a_i}$ is provable in $\mathsf{PV}$). As usual in applications of the KPT Theorem, each function $f_i$ in addition to $1^{(n)}$ might also depend on potential counter-examples to the correctness of the pairs $(n_j,C^{j}_{n_j})$ for $j < i$. In other words, from the provability of $\upci(g_{k'})$ theory $\mathsf{PV}$ proves the universal closure of the following disjunction:\footnote{For simplicity of notation, we left out in each row of (\ref{eq:disjunction}) the condition
$n_i \geq n$, for $i= 1,2, \ldots, r$.}

\begin{multline}\label{eq:disjunction}
\begin{aligned}
&\big[f_1(\un) = (1^{(n_1)}, C_{n_1}^1) \wedge |C^1_{n_1}| \leq n_1^k \wedge (|x_1|=n_1\rightarrow C^1_{n_1}(x_1) = g_{k'}(x_1))\big]\\
\vee &\big[f_2(\un, x_1) = (1^{(n_2)}, C_{n_2}^2) \wedge |C^2_{n_2}| \leq n_2^k \wedge (|x_2|=n_2\rightarrow C^2_{n_2}(x_2) = g_{k'}(x_2))\big] \vee \ldots \\
\vee&\big[f_r(\un, x_1, \ldots, x_{r-1}) = (1^{(n_r)}, C_{n_r}^r) \wedge |C^r_{n_r}| \leq n_r^k \wedge (|x_r|=n_r\rightarrow C^r_{n_r}(x_r) = g_{k'}(x_r))\big].
\end{aligned}
\end{multline}

Modifying the strategy of \citep{unprov}, we argue that either $\mathsf{PV} \nvdash \upci(\widetilde{f_1})$ for a certain  $\mathsf{PV}$ function symbol $\widetilde{f_1}$ that depends on $f_1$ (we are done in this case), or $\mathsf{PV}$ proves that the circuit $C^1_{n_1}$ output by $f_1(1^{(n)})$ does not succeed in computing $g_{k'}$ on inputs of length $n_1 = n_1(1^{(n)})$ for \emph{infinitely many values of} $n$. It will be important that such values of $n$ are polynomially gapped, and that an infinite set $S_1$ of strings of the form $1^{(n)}$ corresponding to them can be enumerated by a $\mathsf{PV}$ function symbol $u_1(1^{(\ell)})$. This allows us to eliminate one disjunct in the sentence obtained from the KPT Theorem if we quantify not over $1^{(n)}$ for all $n$ but just over strings $1^{(n)}$ in the image of $u_1(1^{(\ell)})$,\footnote{As opposed to \citep{unprov}, which focuses on larger input lengths after each iteration of the argument.} since on these specific $1^{(n)}$ the function $f_1$ never succeeds in generating a circuit that correctly computes $g_{k'}$ on input length $n_1 = n_1(1^{(n)})$, and in addition (as we explain below) there is a $\mathsf{PV}$ function symbol that provably produces counter-examples. The proof of our result can be completed by iterating the argument $r$ times while focusing on the relevant input lengths. The idea is similar in spirit to \citep{unprov}, but the argument is more involved because intuitively we need to consider a chain $S_r \subseteq \ldots \subseteq S_1$ of infinite sets of input parameters: If $j \leq i$ then $\mathsf{PV}$ proves that function $f_j$ from the KPT disjunction (with appropriate counter-examples) does not succeed on $1^{(n)} \in S_i$ (assuming the provability of certain auxiliary sentences $\upci(\widetilde{f_j})$). We provide the details next.

Recall that in the terminology of \citep{unprov} the function symbol $\widetilde{f}$ decides $L_\mathsf{succ}$, a padded version of the language $L_{\mathsf{dc}}$ encoding the direct connection language of the circuits generated by $f$. Our definition of $\widetilde{f_1}$ is analogous to the construction in \citep{unprov}, but we need to change the amount of padding in order to accommodate the new setting. Here $f_1(1^{(n)})$ might generate candidate circuits for $g_{k'}$ on larger input lengths. Moreover, the circuits for $\widetilde{f_1}$ obtained from the provability of $\upci(\widetilde{f_1})$ are only guaranteed to work infinitely often. Handling these complications in the case of $f_1$ (and in subsequent cases) will be possible because $g_{k'}$ is hard on \emph{every large enough input length} and the relevant input lengths ($n_1 = n_1(1^{(n)}) \leq n^{a_1}$ in the case of $f_1$) are \emph{provably computable in polynomial time}.

In more detail, let $L^1_\mathsf{dc}$ encode the direct connection language of the sequence of circuits $C^1_{n_1}$ on $n_1 \leq n^{a_1}$ input bits produced by $f_1(1^{(n)})$. Similarly to \citep{unprov}, our language $L^1_\mathsf{succ}$ will be a succinct version of $L^1_\mathsf{dc}$. This time we compress the tuples encoding $C^1_{n_1}$ to:
\[
\langle \mathsf{Bin}(n), 1^{(n^{1/10k})}, u,v,w, 1^{t} \rangle,
\]
where crucially $t$ is arbitrary. (The parameter $t$ is needed in connection to an infinitely often circuit upper bound for $L^1_\mathsf{succ}$, since it makes this language paddable. The use of $t$ here is different than in \citep{unprov}, where it appears only for convenience and as a function of other input parameters.) Under our assumptions, a $p$-time algorithm $\widetilde{f_1}$ deciding $L^1_\mathsf{succ}$ can be defined in $\mathsf{PV}$. Suppose that $\mathsf{PV}  \vdash \upci(\widetilde{f_1})$. Recall that, assuming $C^1_{n_1}$ is a correct circuit for $g_{k'}$, a small circuit for $\widetilde{f_1}$ allows one to obtain a short advice string representing a circuit that decides the tuples of $C^1_{n_1}$, which in turn allows us to compute $g_{k'}$ in time $\ll n_1^{k' - 1}$. Arguing in $\mathsf{PV}$ and adapting the proof of \citep[Lemma 3.2]{unprov} in the natural way (i.e.~by padding $t$ appropriately and using the almost-everywhere hardness of $g_{k'}$), it follows that for infinitely many choices of $1^{(n)}$, $C^1_{n_1}$ does not compute $g_{k'}$ on inputs of length $n_1 = n_1(1^{{(n)}})$. Equivalently,
\[
\mathsf{PV} \vdash \forall 1^{(\ell)}\, \exists 1^{(n)} (n \geq \ell)\, \exists x (|x| = n_1 (1^{(n)})),\; g_{k'}(x) \neq C^1_{|x|}(x)\; .
\]
Using Herbrand's Theorem and in analogy to \citep[Lemma 3.2]{unprov}, there are $\mathsf{PV}$ function symbols $u_1$ and $e_1$ witnessing these existential quantifiers. Furthermore, provably in $\mathsf{PV}$ we have $|u_1(1^{(\ell)})| \leq \ell^{c_1}$ for some constant $c_1$. Therefore, we can take $S_1$ as the infinite set of strings $1^{(n)}$ obtained from $u_1(1^{(\ell)})$ over all choices of $\ell$, and $e_1(1^{(\ell)})$ witnesses that the corresponding circuits $C^1_{n_1}$ are incorrect over the associated input lengths $n_1 = n_1(1^{(n)})$.

The formula obtained from our initial application of the KPT Theorem to $\upci(g_{k'})$ can now be simplified in $\mathsf{PV}$ to a formula equivalent to:
\[
\forall 1^{(n)} \in S_1,\; \text{``KPT disjunct for}~j \in [2,r]~\text{under the counter-example}~x_1 \eqdef e_1(1^{(\ell)})\text{''},
\]
where the quantifier $\forall 1^{(n)} \in S_1$ is  expressed in $\mathsf{PV}$ by ``$\forall 1^{(\ell)} \,\forall 1^{(n)}$ such that $1^{(n)} = u_1(1^{(\ell)})$''. A bit more precisely, the second and later disjuncts in the KPT expression (\ref{eq:disjunction})
     contain functions $f_i$ for $i > 1$ depending on $1^{(n)}$
     and on each $x_j$ for which $j < i$, where the $x_j$ are the variables for counter-examples
     to the correctness of circuits $C^j_{|x_j|}$.
       Now substitute everywhere
          $1^{(n)} = u_1(1^{(\ell)})$  and $x_1 = e_1(1^{(\ell)})$.
     By the choice of $u_1$ and $e_1$,  this substitution provably
     falsifies the first disjunct and also
          $n_1(1^{(n)}) \geq n \geq \ell$.
     Hence (\ref{eq:disjunction}) is turned into a KPT expression with $r-1$ disjuncts, i.e.:
\begin{multline} \label{eq:disjunction2}
\begin{aligned}
&\big [f_2(u_1(1^{(\ell)}), e_1(1^{(\ell)})) = (1^{(n_2)}, C_{n_2}^2) \wedge |C^2_{n_2}| \leq n_2^k
\wedge (|x_2|=n_2\rightarrow C^2_{n_2}(x_2) = g_{k'}(x_2))\big]\\
\vee&\ldots\\
\vee{}&\left[
  \begin{split}
    f_r(u_1(1^{(\ell)}), e_1(1^{(\ell)}),& x_2, \ldots, x_{r-1}) = (1^{(n_r)}, C_{n_r}^r) \wedge |C^r_{n_r}| \leq n_r^k\\
    &\wedge{} (|x_r|=n_r\rightarrow C^r_{n_r}(x_r) = g_{k'}(x_r))
  \end{split}
\right],
\end{aligned}
\end{multline}
where for convenience of notation we have omitted the conditions $n_i \geq \ell$, for $i = 2, \ldots, r$. One can also replace $f_2(u_1(1^{(\ell)}), e_1(1^{(\ell)}))$ by an equivalent term in the language of $\mathsf{PV}$, say, $f'_1(1^{(\ell)})$, and similarly for each $f_i$ appearing in the expression above. We have therefore eliminated one disjunct from the formula appearing in Equation (\ref{eq:disjunction}).

The result is proved as in \citep{unprov} by iterating this argument in the natural way until some derived sentence $\upci(\widetilde{f_i})$ is unprovable or one eliminates all disjuncts. The latter case leads to a contradiction. (Intuitively, the sets $S_i$ contain infinitely many elements and on every string $1^{(n)}$ one of the functions obtained from the initial KPT disjunction must succeed when given appropriate counter-examples. Eliminating all disjuncts contradicts the formula obtained from KPT witnessing, or more precisely, one of the subsequent formulas derived from it in the argument presented above.)

For instance, in the case we get $r = 2$ after the application of the KPT Theorem, assuming that $\mathsf{PV}$ also proves $\upci(\widetilde{f'_1})$, and arguing identically as before, we now get functions $u'_1$ and $e'_1$
computing witnesses (lengths and inputs) that circuits provided by $f'_1$ fail infinitely often to compute $g_{k'}$.\footnote{Complementing our initial informal discussion, while the set $S_1$ is the range of $u_1$, the set $S_2 \subseteq S_1$ is the range of the composed map $u_1 \circ u'_1$.} This is contradictory, because this time the formula from Equation (\ref{eq:disjunction2}) claims that $f'_1$ must succeed (recall that $f'_1(1^{(\ell)}) = f_2(u_1(1^{(\ell)}), e_1(1^{(\ell)}))$) as there are no more disjuncts if $r = 2$.

This completes the proof of Theorem \ref{t:main} item (\emph{c}).\\

\begin{rem}
  We note that in Theorem \ref{t:main} it is possible to \emph{syntactically} enforce the language $L$ to be in the class $\mathcal{C}$ from the formula $\varphi(x)$ and theory $T$. In general, this follows from the definability of these languages in the corresponding theories by formulas of appropriate complexity (see e.g.~\citep[Section 2.6]{buss1997bounded}). In more detail, for part (\emph{a}), as we observed in the concluding remarks of Section \ref{s:intro}, the consistency result extends to theory $\mathsf{S}^2_2(\mathsf{PV})$. In this case, a language in $\mathsf{P}^{\mathsf{NP}}$ is definable in the theory via two provably equivalent $\Sigma^b_2$ and $\Pi^b_2$ formulas. (The provability of the equivalence needs to be done in $\mathsf{S}^2_2(\mathsf{PV})$.) For part (\emph{b}), the corresponding language $L$ is definable in $\mathsf{S}^1_2(\mathsf{PV})$ by a $\Sigma^b_1(\mathsf{PV})$ formula. Lastly, for part (\emph{c}) the proof presented above already implies the claim, since the language is given by a $\mathsf{PV}$ function symbol. Note that in parts (\emph{b}) and (\emph{c}) provability in the theory is not necessary: the syntactic form of the formula (i.e.~$\Sigma^b_1(\mathsf{PV})$ and atomic $\mathsf{PV}$
formula, respectively) imply that the language is in the corresponding class.
\end{rem}

\section*{Acknowledgements}

We would like to thank J\'{a}n Pich, Rahul Santhanam, and Moritz M\"{u}ller for several related discussions. We are also grateful to the reviewers for comments that improved our presentation.

This work  was supported in part by the European Research Council under the European Union's Seventh Framework Programme (FP7/2007-2014)/ERC Grant Agreement no.~615075 and by a Royal Society University Research Fellowship. Jan Byd\v{z}ovsk\'{y} is currently partially supported by the Austrian Science Fund (FWF) under Project P31955.

\bibliographystyle{alpha}
\bibliography{refs-consistency}

\begin{thebibliography}{CMMW19}

\bibitem[AKS02]{Agrawal02primesis}
Manindra Agrawal, Neeraj Kayal, and Nitin Saxena.
\newblock {PRIMES} is in {P}.
\newblock {\em Ann. of Math.}, 2:781--793, 2002.

\bibitem[AW09]{DBLP:journals/toct/AaronsonW09}
Scott Aaronson and Avi Wigderson.
\newblock Algebrization: {A} new barrier in complexity theory.
\newblock {\em {TOCT}}, 1(1):2:1--2:54, 2009.

\bibitem[BFS09]{DBLP:conf/icalp/BuhrmanFS09}
Harry Buhrman, Lance Fortnow, and Rahul Santhanam.
\newblock Unconditional lower bounds against advice.
\newblock In {\em International Colloquium on Automata, Languages and
  Programming \emph{(ICALP)}}, pages 195--209, 2009.

\bibitem[BFT98]{DBLP:conf/coco/BuhrmanFT98}
Harry Buhrman, Lance Fortnow, and Thomas Thierauf.
\newblock Nonrelativizing separations.
\newblock In {\em Conference on Computational Complexity \emph{(CCC)}}, pages
  8--12, 1998.

\bibitem[BH91]{DBLP:journals/iandc/BussH91}
Samuel~R. Buss and Louise Hay.
\newblock On truth-table reducibility to {SAT}.
\newblock {\em Inf. Comput.}, 91(1):86--102, 1991.

\bibitem[BKZ15]{buss2015collapsing}
Samuel~R. Buss, Leszek Ko{\l}odziejczyk, and Konrad Zdanowski.
\newblock Collapsing modular counting in bounded arithmetic and constant depth
  propositional proofs.
\newblock {\em Transactions of the American Mathematical Society},
  367(11):7517--7563, 2015.

\bibitem[BM18]{BydMul}
Jan Byd\v{z}ovsk\'y and Moritz M\"{u}ller.
\newblock Polynomial time ultrapowers and the consistency of circuit lower
  bounds.
\newblock {\em Archive for Mathematical Logic 59, 127–147 (2020)}, 2018.

\bibitem[Bus86]{Bus-book}
Samuel~R. Buss.
\newblock {\em Bounded arithmetic}, volume~86.
\newblock Bibliopolis, 1986.

\bibitem[Bus97]{buss1997bounded}
Samuel~R. Buss.
\newblock Bounded arithmetic and propositional proof complexity.
\newblock In {\em Logic of computation}, pages 67--121. Springer, 1997.

\bibitem[Cai07]{DBLP:journals/jcss/Cai07}
Jin{-}yi Cai.
\newblock $\mathsf{S}_2^p$ is a subset of $\mathsf{ZPP}^\mathsf{NP}$.
\newblock {\em J. Comput. Syst. Sci.}, 73(1):25--35, 2007.

\bibitem[CC06]{DBLP:journals/jco/CaiC06}
Jin{-}yi Cai and Venkatesan~T. Chakaravarthy.
\newblock On zero error algorithms having oracle access to one query.
\newblock {\em J. Comb. Optim.}, 11(2):189--202, 2006.

\bibitem[CK07]{CooKra}
Stephen~A. Cook and Jan Kraj{\'{\i}}cek.
\newblock Consequences of the provability of \emph{NP} {\(\subseteq\)}
  \emph{P/poly}.
\newblock {\em J. Symb. Log.}, 72(4):1353--1371, 2007.

\bibitem[CMMW19]{CMMW19}
Lijie Chen, Dylan~M. McKay, Cody~D. Murray, and Ryan Williams.
\newblock Relations and equivalences between circuit lower bounds and
  {K}arp-{L}ipton theorems.
\newblock In {\em Computational Complexity Conference \emph{(CCC)}}, 2019.

\bibitem[CN10]{CooNgu}
Stephen Cook and Phuong Nguyen.
\newblock {\em Logical foundations of proof complexity}, volume~11.
\newblock Cambridge University Press Cambridge, 2010.

\bibitem[Cob65]{Cob64}
Alan Cobham.
\newblock The intrinsic computational difficulty of functions.
\newblock {\em Proc. Logic, Methodology and Philosophy of Science}, pages
  24--30, 1965.

\bibitem[Coo75]{Coo75}
Stephen~A. Cook.
\newblock Feasibly constructive proofs and the propositional calculus
  (preliminary version).
\newblock In {\em Symposium on Theory of Computing \emph{(STOC)}}, pages
  83--97, 1975.

\bibitem[CP08]{DBLP:conf/coco/ChangP08}
Richard Chang and Suresh Purini.
\newblock Amplifying $\mathsf{ZPP}^{\mathsf{sat}[1]}$ and the two queries
  problem.
\newblock In {\em Conference on Computational Complexity \emph{(CCC)}}, pages
  41--52, 2008.

\bibitem[DPV18]{DBLP:conf/mfcs/DixonPV18}
Peter Dixon, Aduri Pavan, and N.~V. Vinodchandran.
\newblock On pseudodeterministic approximation algorithms.
\newblock In {\em Symposium on Mathematical Foundations of Computer Science
  \emph{(MFCS)}}, pages 61:1--61:11, 2018.

\bibitem[FGHK16]{DBLP:conf/focs/FindGHK16}
Magnus~Gausdal Find, Alexander Golovnev, Edward~A. Hirsch, and Alexander~S.
  Kulikov.
\newblock A better-than-$3n$ lower bound for the circuit complexity of an
  explicit function.
\newblock In {\em Symposium on Foundations of Computer Science \emph{(FOCS)},
  {USA}}, pages 89--98, 2016.

\bibitem[FS17]{DBLP:journals/iandc/FortnowS17}
Lance Fortnow and Rahul Santhanam.
\newblock Robust simulations and significant separations.
\newblock {\em Inf. Comput.}, 256:149--159, 2017.

\bibitem[FSW09]{FSW09}
Lance Fortnow, Rahul Santhanam, and Ryan Williams.
\newblock Fixed-polynomial size circuit bounds.
\newblock In {\em Conference on Computational Complexity \emph{(CCC)}}, pages
  19--26, 2009.

\bibitem[GZ11]{DBLP:books/sp/goldreich2011/GoldreichZ11}
Oded Goldreich and David Zuckerman.
\newblock Another proof that $\mathsf{BPP} \subseteq \mathsf{PH}$ (and more).
\newblock In {\em Studies in Complexity and Cryptography}, pages 40--53. 2011.

\bibitem[Hem89]{DBLP:journals/jcss/Hemachandra89}
Lane~A. Hemachandra.
\newblock The strong exponential hierarchy collapses.
\newblock {\em J. Comput. Syst. Sci.}, 39(3):299--322, 1989.

\bibitem[Je{\v r}04]{Jer04}
Emil Je{\v r}{\'a}bek.
\newblock Dual weak pigeonhole principle, boolean complexity, and
  derandomization.
\newblock {\em Ann. Pure Appl. Logic}, 129(1-3):1--37, 2004.

\bibitem[Je{\v r}05]{Jer-phd}
Emil Je{\v r}{\'a}bek.
\newblock Weak pigeonhole principle, and randomized computation.
\newblock {\em Ph.D. {T}hesis, Charles University in Prague}, 2005.

\bibitem[Je{\v r}06]{jerabek:sharply-bounded}
Emil Je{\v r}{\'a}bek.
\newblock The strength of sharply bounded induction.
\newblock {\em Mathematical Logic Quarterly}, 52(6):613--624, 2006.

\bibitem[Je{\v r}07]{Jer07}
Emil Je{\v r}{\'a}bek.
\newblock Approximate counting in bounded arithmetic.
\newblock {\em J. Symb. Log.}, 72(3):959--993, 2007.

\bibitem[Je{\v r}09]{Jer09}
Emil Je{\v r}{\'a}bek.
\newblock Approximate counting by hashing in bounded arithmetic.
\newblock {\em J. Symb. Log.}, 74(3):829--860, 2009.

\bibitem[Juk12]{DBLP:books/daglib/0028687}
Stasys Jukna.
\newblock {\em Boolean Function Complexity - Advances and Frontiers}.
\newblock Springer, 2012.

\bibitem[Kan82]{Kan}
Ravi Kannan.
\newblock Circuit-size lower bounds and non-reducibility to sparse sets.
\newblock {\em Information and Control}, 55(1-3):40--56, 1982.

\bibitem[KO17]{unprov}
Jan Kraj{\'{\i}}{\v c}ek and Igor~Carboni Oliveira.
\newblock Unprovability of circuit upper bounds in {C}ook's theory {PV}.
\newblock {\em Logical Methods in Computer Science}, 13(1), 2017.

\bibitem[KPT91]{KPT}
Jan Kraj{\'{\i}}{\v c}ek, Pavel Pudl{\'{a}}k, and Gaisi Takeuti.
\newblock Bounded arithmetic and the polynomial hierarchy.
\newblock {\em Ann. Pure Appl. Logic}, 52(1-2):143--153, 1991.

\bibitem[Kra93]{fragm}
Jan Kraj{\'\i}{\v{c}}ek.
\newblock Fragments of bounded arithmetic and bounded query classes.
\newblock {\em Transactions of the American Mathematical Society},
  338(2):587--598, 1993.

\bibitem[Kra95]{kniha}
Jan Kraj\'{i}{\v c}ek.
\newblock {\em Bounded Arithmetic, Propositional Logic, and Complexity Theory}.
\newblock Cambridge University Press, 1995.

\bibitem[Kra98]{krajicek_extensions}
Jan Kraj{\'\i}{\v{c}}ek.
\newblock Extensions of models of {PV}.
\newblock {\em Lecture Notes in Logic}, 11:104--114, 1998.

\bibitem[KW98]{DBLP:journals/siamcomp/KoblerW98}
Johannes K{\"{o}}bler and Osamu Watanabe.
\newblock New collapse consequences of {NP} having small circuits.
\newblock {\em {SIAM} J. Comput.}, 28(1):311--324, 1998.

\bibitem[LC11]{DBLP:journals/corr/abs-1103-5215}
Dai Tri~Man Le and Stephen~A. Cook.
\newblock Formalizing randomized matching algorithms.
\newblock {\em Logical Methods in Computer Science}, 8(3), 2011.

\bibitem[Le14]{DTML_thesis}
Dai Tri~Man Le.
\newblock Bounded arithmetic and formalizing probabilistic proofs.
\newblock {\em Ph.D. Thesis, University of Toronto}, 2014.

\bibitem[Lip94]{DBLP:conf/coco/Lipton94}
Richard~J. Lipton.
\newblock Some consequences of our failure to prove non-linear lower bounds on
  explicit functions.
\newblock In {\em Structure in Complexity Theory Conference \emph{(CCC)}},
  pages 79--87, 1994.

\bibitem[MP17]{DBLP:journals/eccc/MullerP17}
Moritz M{\"{u}}ller and J{\'{a}}n Pich.
\newblock Feasibly constructive proofs of succinct weak circuit lower bounds.
\newblock {\em Electronic Colloquium on Computational Complexity
  \emph{(ECCC)}}, 24:144, 2017.

\bibitem[Mul99]{DBLP:journals/siamcomp/Mulmuley99}
Ketan Mulmuley.
\newblock Lower bounds in a parallel model without bit operations.
\newblock {\em {SIAM} J. Comput.}, 28(4):1460--1509, 1999.

\bibitem[MW18]{DBLP:conf/stoc/MurrayW18}
Cody Murray and R.~Ryan Williams.
\newblock Circuit lower bounds for nondeterministic quasi-polytime: an easy
  witness lemma for {NP} and {NQP}.
\newblock In {\em Symposium on Theory of Computing \emph{(STOC)}}, pages
  890--901, 2018.

\bibitem[Oja04]{ojakian2004combinatorics}
Kerry Ojakian.
\newblock Combinatorics in bounded arithmetic.
\newblock {\em Ph.D. Thesis, Carnegie Mellon University}, 2004.

\bibitem[OS17]{DBLP:conf/stoc/OliveiraS17}
Igor~Carboni Oliveira and Rahul Santhanam.
\newblock Pseudodeterministic constructions in subexponential time.
\newblock In {\em Symposium on Theory of Computing \emph{(STOC)}}, pages
  665--677, 2017.

\bibitem[Pic14]{Pich-phd}
J{\'{a}}n Pich.
\newblock Complexity theory in feasible mathematics.
\newblock {\em Ph.D. Thesis, Charles University in Prague}, 2014.

\bibitem[Pic15a]{Pich}
J{\'{a}}n Pich.
\newblock Circuit lower bounds in bounded arithmetics.
\newblock {\em Ann. Pure Appl. Logic}, 166(1):29--45, 2015.

\bibitem[Pic15b]{Pich-pcp}
J{\'{a}}n Pich.
\newblock Logical strength of complexity theory and a formalization of the
  {PCP} theorem in bounded arithmetic.
\newblock {\em Logical Methods in Computer Science}, 11(2), 2015.

\bibitem[Pud13]{DBLP:books/daglib/0031527}
Pavel Pudl{\'{a}}k.
\newblock {\em Logical Foundations of Mathematics and Computational Complexity
  - {A} Gentle Introduction}.
\newblock Springer, 2013.

\bibitem[Raz95]{razborov1995bounded}
Alexander~A. Razborov.
\newblock Bounded arithmetic and lower bounds in boolean complexity.
\newblock In {\em Feasible Mathematics II}, pages 344--386. Springer, 1995.

\bibitem[Ros10]{rossmanphdthesis}
Benjamin Rossman.
\newblock Average-case complexity of detecting cliques.
\newblock {\em Ph.D. Thesis, MIT}, 2010.

\bibitem[San09]{DBLP:journals/siamcomp/Santhanam09}
Rahul Santhanam.
\newblock Circuit lower bounds for {M}erlin-{A}rthur classes.
\newblock {\em {SIAM} J. Comput.}, 39(3):1038--1061, 2009.

\bibitem[SW14]{SW}
Rahul Santhanam and Ryan Williams.
\newblock On uniformity and circuit lower bounds.
\newblock {\em Computational Complexity}, 23(2):177--205, 2014.

\bibitem[Vin05]{DBLP:journals/tcs/Vinodchandran05}
N.~V. Vinodchandran.
\newblock A note on the circuit complexity of {PP}.
\newblock {\em Theor. Comput. Sci.}, 347(1-2):415--418, 2005.

\bibitem[WP87]{PW}
A.~J. Wilkie and Jeff~B. Paris.
\newblock On the scheme of induction for bounded arithmetic formulas.
\newblock {\em Ann. Pure Appl. Logic}, 35:261--302, 1987.

\end{thebibliography}

\appendix

\section{Consistency of $\mathsf{P} \neq \mathsf{NP}$ from unprovability of lower bounds}\label{s:appendix_unprov}

Imagine that against most expectations $\mathsf{P}$ is actually equal to $\mathsf{NP}$ and there is a polynomial time algorithm $f$ (i.e.~a $\mathsf{PV}$ function symbol)
that finds a satisfying assignment for all satisfiable formulas. In other words, if $\psi_\mathsf{SAT}(x,y)$ denotes an $L(\mathsf{PV})$-formula
that checks if $y$ satisfies the formula encoded by $x$, then the sentence
\begin{equation} \label{31.1.19a}
\pnp(f)  \;\eqdef\; \forall x\, \forall y\,[\psi_\mathsf{SAT}(x,y) \rightarrow \psi_\mathsf{SAT}(x, f(x))]
\end{equation}
is true in the standard model.
Now suppose that in order to prove the universal statement $\pnp(f)$ in Equation (\ref{31.1.19a}) you have to use concepts (definitions, predicates, etc.) that cannot be defined
as polynomial-time algorithms. To be more specific, assume that (\ref{31.1.19a}) is provable using induction for non-deterministic
polynomial-time algorithms (corresponding to theory $\mathsf{T}^1_2(\mathsf{PV})$), but not using induction for polynomial-time algorithms only (corresponding to theory $\mathsf{PV}$). Could we still maintain that the mere existence of $f$ implies that the satisfiability problem is ``feasible''?

This question is more philosophical than mathematical, and we are not going to offer an answer. Instead, we
suggest to consider a strictly mathematical question.

\begin{conj} \label{31.1.19b}
For no polynomial-time algorithm $f$  theory $\mathsf{PV}$ proves the sentence $\pnp(f)$.
\end{conj}

Informally, Conjecture \ref{31.1.19b} states that $\mathsf{PV}$ and by standard conservation results  $\mathsf{S}^1_2$ are both consistent with $\mathsf{P} \neq \mathsf{NP}$.
That is, either $\mathsf{P} \neq \mathsf{NP}$ as often assumed, and hence the conjecture is trivially true,
or $\mathsf{P} = \mathsf{NP}$ but you cannot prove it using only polynomial-time concepts and reasoning. For this reason, Conjecture \ref{31.1.19b}
is a formal weakening of the conjecture that $\mathsf{P} \neq \mathsf{NP}$.

We do not claim any originality for the conjecture; not only it follows from $\mathsf{P} \neq \mathsf{NP}$ but the statement is also known to follow from the conjectures that bounded arithmetic does not collapse to $\mathsf{PV}$
or that the Extended Frege propositional proof system is not polynomially bounded. The conjecture must have been also
one of the ideas leading Stephen Cook to his seminal paper \cite{Coo75}. We think it is a weakening of the
$\mathsf{P}$ vs.~$\mathsf{NP}$ conjecture that has an intrinsic relevance to it, and that it ought to be studied more
(cf.~\cite{CooKra} for more discussion).\\

In this appendix, we observe that Conjecture \ref{31.1.19b} is related to the \emph{unprovability of circuit lower bounds}. For a $\mathsf{PV}$ function symbol $h$ and a circuit size parameter $k \in \mathbb{N}$, consider the sentence
\begin{equation}\label{eq:LB}
\LB(h) \;\eqdef\; \neg \upci(h)\; ,
\end{equation}
where $\upci(h)$ is the sentence from Equation (\ref{eq:defformula}). Intuitively, $\LB(h)$ states that the language defined by $h$ is hard on input length $m$ for circuits of size $m^k$ whenever $m \geq n$, for a fixed value $n$.

\begin{thm}[Consistency of lower bounds with $\mathsf{PV}$ from the unprovability \nolinebreak of \nolinebreak lower \nolinebreak bounds] \label{t:thm_appendix}
If there exists $k \in \mathbb{N}$ such that for no function symbol $h$  theory $\mathsf{PV}$ proves the sentence $\LB(h)$, then Conjecture \ref{31.1.19b} holds.
\end{thm}

Note that the hypothesis of Theorem \ref{t:thm_appendix} is weaker than the assumption that $\mathsf{PV}$ does not prove that $\mathsf{NP} \nsubseteq \mathsf{SIZE}[n^k]$ for some $k$. Roughly speaking, Theorem \ref{t:thm_appendix} shows that if $\mathsf{PV}$ does not prove circuit lower bounds then $\mathsf{P} \neq \mathsf{NP}$ is consistent with $\mathsf{PV}$.

\begin{proof}[Sketch of the proof of Theorem \emph{\ref{t:thm_appendix}}] The argument proceeds in the contrapositive.
We formalize in $\mathsf{PV}$ the result that if $\mathsf{P} = \mathsf{NP}$ then for each parameter $k$, $\mathsf{P} \nsubseteq \textit{i.o.}\mathsf{SIZE}[n^k]$ (see e.g.~\citep[Theorem 3]{DBLP:conf/coco/Lipton94}). Recall that this is obtained by combining the collapse of $\mathsf{PH}$ to $\mathsf{P}$ together with Kannan's argument \citep{Kan} showing that $\mathsf{PH}$ can define languages that are almost-everywhere hard against circuits of fixed-polynomial size. The usual proof of this claim shows via a counting argument the existence of a truth-table of size $2^n$ that is hard against circuit size $n^k$. A potential issue is that this result might not be available in $\mathsf{PV}$.

We overcome this difficulty as follows. From the provability in $\mathsf{PV}$ that $\mathsf{P} = \mathsf{NP}$, it follows that the hierarchy $\mathsf{T_2}(\mathsf{PV})$ of bounded arithmetic theories $\mathsf{T}^i_2(\mathsf{PV})$ collapses to $\mathsf{PV}$ \citep{KPT}. Recall that the surjective weak pigeonhole principle $\mathsf{sWPHP}$ for $\mathsf{PV}$ function symbols is provable in $\mathsf{T}^2_2(\mathsf{PV})$ (see e.g.~\citep{kniha}). Define a $\mathsf{PV}$ function symbol $g$ that takes as input a circuit $C$ of size $n^k$ and outputs the first $n^{k + 1}$ bits of the truth-table computed by $C$. From $\mathsf{sWPHP}(g)$ we now derive in $\mathsf{PV}$ that the prefix of some truth-table is not computable by circuits of size $n^k$, if $n$ is sufficiently large. We can (implicitly) extend the lexicographic first truth-table prefix satisfying this property with zeroes, and use the resulting truth-table to define a $\mathsf{PV}$-formula $\varphi(x)$ with a constant number of bounded quantifiers that defines a language $L$ that is hard against circuits of size $n^k$, where the hardness is provable in $\mathsf{PV}$. Since the provability in $\mathsf{PV}$ that $\mathsf{P} = \mathsf{NP}$ implies the provability in $\mathsf{PV}$ that $\mathsf{PH}$ collapses to $\mathsf{P}$, it follows that $\varphi(x)$ is equivalent in $\mathsf{PV}$ to the language defined by some $\mathsf{PV}$ function symbol $h$. In other words, $\mathsf{PV} \vdash \LB(h)$, which completes the proof of Theorem \ref{t:thm_appendix}. \end{proof}

\end{document}